\documentclass{article}

\usepackage{authblk}
\usepackage{amsmath}
\usepackage{amssymb}
\usepackage{amsthm}

\title{\bf Extending Newton's Laws of Motion with Free Will}

\author{Giovanni Giuffrida}
\affil{University of Catania \\ \texttt{giovanni.giuffrida@unict.it}}

\author{Calogero G. Zarba}
\affil{Neodata Group s.r.l. \\ \texttt{calogero.zarba@neodatagroup.com}}

\date{}
\theoremstyle{plain}
\newtheorem{theorem}{Theorem}

\theoremstyle{definition}
\newtheorem{definition}{Definition}
\newtheorem{example}{Example}
\newtheorem{conjecture}{Conjecture}

\begin{document}
\maketitle

\begin{abstract}
We study the concept of free will by defining a mathematical model that extends Newton's laws of motion, in such a way that bodies are replaced with agents endowed with free will.
In our model the free will of agents
is not entirely free, but is bound by the Golden Rule, an 
ethic found in almost all cultures, which states that one should wish upon others as one wishes upon himself.
\end{abstract}

% --------------------
\section{Introduction}
Free will~\cite{Kan2005,Lan1968,Omo2015} is defined as the ability of agents to choose among a set of possible actions.

Free will has implications in religion, ethics, and science.
In religion, the existence of free will limits the power of deities.
In ethics, the existence of free will implies that agents are morally responsible for their own actions.
In science, the existence of free will means that the law of nature cannot be completely deterministic, but must allow for a certain degree of non-random indeterminism.

In this paper we study the concept of free will by defining a mathematical model that extends Newton's laws of motion~\cite{TM2003}, 
in such a way that bodies are replaced with agents endowed with free will. 
Although historically Newton's laws of motion have been linked to determinism and a clockwork universe,
we show that Newton's laws are consistent with indeterminism and therefore with a universe in which free will is present.
In fact, while it is true that in Newton's original treatment the forces between bodies are deterministic and depend solely on the past,
in this paper the forces between agents are not deterministic and are instead chosen by agents by virtue of free will.

We first define our model assuming that time is discrete, and then we extend it to the case of continuous time.
In both the discrete and continuous cases, we view agents as points in a three-dimensional space with a fixed position and velocity for each instant of time.
Agents extert on each others forces that modify their positions and velocities in a deterministic way.
However, the forces exterted by the agents depend on free will, that is, they depend on choices made by the agents.  These choices are
not completely free, but they are bound by the Golden Rule, an ethic found in almost all cultures, 
which can be stated in various forms:
\begin{itemize}
\item \emph{treat others as you would like others to treat you} (positive or directive form);

\item \emph{do not treat others in ways that you would not like to be treated} (negative or prohibitive form);

\item \emph{what you wish upon others, you wish upon yourself} (empathetic or responsive form).
\end{itemize}
In our model, the Golden Rule takes the following form: whenever an agent exterts a force on another agent, it also must extert upon itself a force of opposite direction and equal intensity.

In our model of free will, agents choose their forces gradually, as time goes by.  
Precisely, the forces exterted by the agents at instant of time $t$ are chosen immediately 
after observing their positions and velocities at instant of time $t$.

Our model of free will is non-local, in the sense that agents can extert forces on each other no matter what is the distance between them.
To reduce non-locality, we introduce and discuss the concept of bounded free will, in which the forces that agents
extert on each other are limited by the inverse of the distance between them.

The rest of this paper is organized as follows.  
Section~\ref{se:related} reviews related work.
Section~\ref{se:discrete} rigorously defines our mathematical model of free will in the case that time is assumed to be discrete.
Section~\ref{se:theorems-discrete} formally proves some properties of our discrete time model of free will.
Section~\ref{se:continuous} extends our discrete time model of free will to the case in which time is continuous.
Section~\ref{se:theorems-continuous} formally proves some properties of our continuous time model of free will.
Section~\ref{se:newton} compares our mathematical model of free will with Newton's laws of motion.
Section~\ref{se:timing} discusses the subtle point of when agents make their choices.
Section~\ref{se:bounded} introduces and discuss the concept of bounded free will.
Section~\ref{se:conclusion} concludes the paper with final remarks and hints for future work.

% --------------------
\section{Related Work}
\label{se:related}

The concept of free will is connected with the concept of determinism~\cite{Lap2018,Smi2000}, 
the view that all events are completely determined by previous causes.  
According to incompatibilists~\cite{Kan1989,Van1975}, determinism suggests
that there is only one possible course of action, and therefore impedes the possibility of agents to make choices.
On the other hand, there are compatibilists~\cite{Hum2018} which hold the view that 
determinism is not a threat to free will, and therefore the two of them are compatible.

The concept of free will has been studied in quantum mechanics, where the free will theorem~\cite{CK2006} states that, if people
have free will then, under certain assumptions, so must some elementary particles.

In biology~\cite{PS2017}, the concept of free will is at the core of the nature vs.\ nurture debate, the problem of 
whether human behavior is determined by the environment or by a person's genes.

In the neurosciences~\cite{Lib1985}, the concept of free will is analyzed by experiments on human brain activity
which try to establish at which time people become aware of their actions.

This paper is the second attempt of the authors to mathematically model the concept of free will and its relationship with the Golden Rule.
In the first attempt~\cite{GZ2013}, the authors used social network analysis, and defined a model of a society in which persons are
endowed with free will.  According to this social model, if a person $A$ wishes upon another person $B$, then in the future person $B$ must
wish in the same way upon person $A$.  
The authors now believe this is an ethically wrong view of the Golden Rule because, if an action made by a person to another
person is ethically wrong, repeating it with roles reversed does not make it right.  
Furthermore, there is an unnecessary restriction on the free will of persons.  In other words, if you
perform an action on me, I should have the power not to respond with the same currency.

% ---------------------------------------
\section{Discrete time model of free will}
\label{se:discrete}

\begin{definition}[Free will system]
A \emph{free will system} is a finite non-empty set $\{a_1, \ldots, a_n\}$ of agents, together with a map that assigns to each agent $a_i$
\begin{itemize}
\item an individual mass $m_i \in \mathbb{R}^+$;

\item an individual initial position $\mathbf{x}_i^{(0)} \in \mathbb{R}^3$;

\item an individual initial velocity $\mathbf{v}_i^{(0)} \in \mathbb{R}^3$.
\end{itemize}
\end{definition}

\begin{example}
\label{ex:system}
A free will system $S = \{a_1, a_2\}$ with only two agents may be specified by letting
\begin{align*}
  m_1 &= 1                         & m_2 &= 1 \\
  \mathbf{x}_1^{(0)} &= (-1, 0, 0) &\mathbf{x}_2^{(0)} &= (1, 0, 0) \\
  \mathbf{v}_1^{(0)} &= (0, 0, 0)  &\mathbf{v}_2^{(0)} &= (0, 0, 0)
\end{align*}
\end{example}

\begin{definition}[Discrete individual will]
\label{de:will}
Let $S$ be a free will system whose set of agents is $\{a_1, ..., a_n\}$.
A \emph{discrete individual will} of agent $a_i$ in $S$ is a set of functions
\begin{equation*}
  \left \{ \mathbf{F}_{ij} : \mathbb{N} \to \mathbb{R}^3 \mid j \neq i \right\} \,.
\end{equation*}
\end{definition}

Intuitively, in the definition of discrete individual will, $\mathbf{F}_{ij}(t)$ is the force exterted by agent $a_i$ upon agent $a_j$ at
instant of time $t$, determined by the free will of agent $a_i$.

\begin{definition}[Discrete collective will]
Let $S$ be a free will system whose set of agents is $\{a_1, ..., a_n\}$.
A \emph{discrete collective will} for $S$ is the union $\bigcup_i w_i$, where $w_i$ is a discrete individual will of agent $a_i$ in $S$, for all $i$.
\end{definition}

\begin{example}
\label{ex:will}
Let $S = \{a_1, a_2\}$ be the free will system specified in Example~\ref{ex:system}.  
Let us informally assume that the wish of the two agents is to stay as far apart from each other as possible, and that
the wish of agent $a_1$ is stronger than the wish of agent $a_2$.  Mathematically, we specify the following collective will
\begin{align*}
  \mathbf{F}_{12}(t) &= (2, 0, 0) \,, &\text{for all } t \ge 0 \,, \\
  \mathbf{F}_{21}(t) &= (-1, 0, 0)\,, &\text{for all } t \ge 0 \,.
\end{align*}
\end{example}

\begin{definition}[Discrete history]
Let $S$ be a free will system whose set of agents is $\{a_1, ..., a_n\}$, 
and let $W = \left\{ \mathbf{F}_{ij} : \mathbb{N} \to \mathbb{R}^3 \mid i \neq j \right\}$ be a discrete collective will for $S$.
The \emph{discrete history} of $S$ with respect to $W$ is the following collection of functions:
\begin{itemize}
\item $\mathbf{F}_i : \mathbb{N} \to \mathbb{R}^3$, for each $i$, where $\mathbf{F}_i(t)$ is the total force exterted on agent $a_i$ at instant of time $t$;

\item $\mathbf{a}_i : \mathbb{N} \to \mathbb{R}^3$, for each $i$, where $\mathbf{a}_i(t)$ is the acceleration of agent $a_i$ at instant of time $t$;

\item $\mathbf{v}_i : \mathbb{N} \to \mathbb{R}^3$, for each $i$, where $\mathbf{v}_i(t)$ is the velocity of agent $a_i$ at instant of time $t$;

\item $\mathbf{x}_i : \mathbb{N} \to \mathbb{R}^3$, for each $i$, where $\mathbf{x}_i(t)$ is the position of agent $a_i$ at instant of time $t$.
\end{itemize}
For each $i$, the above functions are precisely defined as follows:
\begin{align}
\label{eq:total}
  \mathbf{F}_i(t) = \sum_{j \neq i} \left( \mathbf{F}_{ji}(t) - \mathbf{F}_{ij}(t) \right) \,, &&\text{for all } t \in \mathbb{N} \,,
\end{align}

\begin{align} 
\label{eq:second}
  \mathbf{a}_i(t) = \frac{\mathbf{F}_i(t)}{m_i} \,, &&\text{for all } t \in \mathbb{N} \,,
\end{align}

\begin{align*}
  \mathbf{v}_i(0) = \mathbf{v}_i^{(0)} \,, &&\phantom{'}
\end{align*}

\begin{align}
\label{eq:velocity}
  \mathbf{v}_i(t + 1) = \mathbf{v}_i(t) + \mathbf{a}_i(t) \,, &&\text{for all } t \in \mathbb{N} \,,
\end{align}

\begin{align*}
  \mathbf{x}_i(0) = \mathbf{x}_i^{(0)} \,, &&\phantom{'}
\end{align*}

\begin{align}
\label{eq:position}
  \mathbf{x}_i(t + 1) = \mathbf{x}_i(t) + \mathbf{v}_i(t + 1) \,, &&\text{for all } t \in \mathbb{N} \,.
\end{align}
\end{definition}

Notice that equation~\eqref{eq:total} enforces agents to obey the Golden Rule.  
In fact, at each instant of time $t$, an agent $a_i$ that chooses to extert a force $\mathbf{F}_{ij}(t)$ upon agent $a_j$, must also 
extert the force $-\mathbf{F}_{ij}(t)$ upon itself.

\begin{example}
Let $S = \{a_1, a_2\}$ be the free will system specified in Example~\ref{ex:system}, and let $W$ be the 
collective will specified in Example~\ref{ex:will}.  
Using a program written in Java, we run a simulation of the discrete history of $S$ with respect to $W$.
The output of the simulation, executed for three iterations, is as follows
\begin{verbatim}
  t = 0
    x[1] = (-1.0, 0.0, 0.0)
    v[1] = (0.0, 0.0, 0.0)
    x[2] = (1.0, 0.0, 0.0)
    v[2] = (0.0, 0.0, 0.0)
\end{verbatim}

\begin{verbatim}
  t = 1
    x[1] = (-4.0, 0.0, 0.0)
    v[1] = (-3.0, 0.0, 0.0)
    x[2] = (4.0, 0.0, 0.0)
    v[2] = (3.0, 0.0, 0.0)
\end{verbatim}

\begin{verbatim}
  t = 2
    x[1] = (-10.0, 0.0, 0.0)
    v[1] = (-6.0, 0.0, 0.0)
    x[2] = (10.0, 0.0, 0.0)
    v[2] = (6.0, 0.0, 0.0)
\end{verbatim}

\begin{verbatim}
  t = 3
    x[1] = (-19.0, 0.0, 0.0)
    v[1] = (-9.0, 0.0, 0.0)
    x[2] = (19.0, 0.0, 0.0)
    v[2] = (9.0, 0.0, 0.0)
\end{verbatim}
\end{example}

% ------------------------------------------------------
\section{Discrete time properties}
\label{se:theorems-discrete}

\begin{theorem}[Balance of forces]
\label{th:zero}
Let $S$ be a free will system whose set of agents is $\{a_1, ..., a_n\}$, 
let $W$ be a discrete collective will for $S$,
and let $H$ be the discrete history of $S$ with respect to $W$.
Then, for each $i$, the following holds
\begin{align*}
  \sum_i \mathbf{F}_i(t) = 0 \,, &&\text{for all } t \in \mathbb{N} \,.
\end{align*}
\end{theorem}
\begin{proof}
Using equation~\eqref{eq:total} and the commutative property of sums, we have
\begin{align*}
  \sum_i \mathbf{F}_i(t) &= \sum_i \sum_{j \neq i} \left( \mathbf{F}_{ji}(t) - \mathbf{F}_{ij}(t) \right) \\
  &= \sum_i \sum_{j \neq i} \mathbf{F}_{ji}(t) -  \sum_i \sum_{j \neq i} \mathbf{F}_{ij}(t) \\
  &= \sum_i \sum_{j \neq i} \mathbf{F}_{ij}(t) -  \sum_i \sum_{j \neq i} \mathbf{F}_{ij}(t) \\
  &= 0 \,.
\end{align*}
\end{proof}

\begin{definition}[Individual momentum]
Let $S$ be a free will system whose set of agents is $\{a_1, ..., a_n\}$, 
let $W$ be a discrete collective will for $S$,
and let $H$ be the discrete history of $S$ with respect to $W$.
The \emph{individual momentum} $\mathbf{p}_i(t)$ of agent $a_i$ at instant of time $t$ is
\begin{equation*}
  \mathbf{p}_i(t) = m_i \mathbf{v}_i(t) \,.
\end{equation*}
\end{definition}

\begin{definition}[Collective momentum]
Let $S$ be a free will system whose set of agents is $\{a_1, ..., a_n\}$, 
let $W$ be a discrete collective will for $S$,
and let $H$ be the discrete history of $S$ with respect to $W$.
The \emph{collective momentum} $\mathbf{p}_S(t)$ of $S$ at instant of time $t$ is
\begin{equation*}
  \mathbf{p}_S(t) = \sum_i \mathbf{p}_i(t) \,.
\end{equation*}
\end{definition}

\begin{theorem}[Conservation of collective momentum]
Let $S$ be a free will system whose set of agents is $\{a_1, ..., a_n\}$, 
let $W$ be a discrete collective will for $S$,
and let $H$ be the discrete history of $S$ with respect to $W$.
Then there exists a constant $\mathbf{p}_S \in \mathbb{R}^3$ such that
\begin{align*}
  \mathbf{p}_S = \mathbf{p}_S(t) \,, &&\text{for all } t \in \mathbb{N} \,.
\end{align*}
\end{theorem}
\begin{proof}
It is enough to prove, by induction on $t$, that $\mathbf{p}_S(t) = \mathbf{p}_S(0)$.
The base case is trivial.  For the inductive step, notice that by equation~\eqref{eq:velocity} we have
\begin{equation*}
  \mathbf{v}_i(t + 1) = \mathbf{v}_i(t) + \mathbf{a}_i(t) \,.
\end{equation*}
By equation~\eqref{eq:second} it follows that
\begin{equation*}
  \mathbf{v}_i(t + 1) = \mathbf{v}_i(t) + \frac{\mathbf{F}_i(t)}{m_i} \,,
\end{equation*}
which implies
\begin{equation*}
  m_i \mathbf{v}_i(t + 1) = m_i \mathbf{v}_i(t) + \mathbf{F}_i(t) \,,
\end{equation*}
But then
\begin{equation*}
  \sum_i m_i \mathbf{v}_i(t + 1) = \sum_i m_i \mathbf{v}_i(t) + \sum_i \mathbf{F}_i(t) \,,
\end{equation*}
and by Theorem~\ref{th:zero}, it follows
\begin{equation*}
  \sum_i m_i \mathbf{v}_i(t + 1) = \sum_i m_i \mathbf{v}_i(t) \,,
\end{equation*}
which is equivalent to
\begin{equation*}
  \mathbf{p}_S(t + 1) = \mathbf{p}_S(t) \,.
\end{equation*}
\end{proof}

\begin{definition}[Collective position]
Let $S$ be a free will system whose set of agents is $\{a_1, ..., a_n\}$, 
let $W$ be a discrete collective will for $S$,
and let $H$ be the discrete history of $S$ with respect to $W$.
The \emph{collective position}\footnote{Also known in literature as center of mass.} $\mathbf{x}_S(t)$ of $S$ at instant of time $t$ is
\begin{equation*}
  \mathbf{x}_S(t) = \frac{\sum_i m_i \mathbf{x}_i(t)}{\sum_i m_i} \,.
\end{equation*}
\end{definition}

\begin{definition}[Collective velocity]
Let $S$ be a free will system whose set of agents is $\{a_1, ..., a_n\}$, 
let $W$ be a discrete collective will for $S$,
and let $H$ be the discrete history of $S$ with respect to $W$.
The \emph{collective velocity} $\mathbf{v}_S(t)$ of $S$ at instant of time $t$ is
\begin{equation*}
  \mathbf{v}_S(t) = \frac{\sum_i m_i \mathbf{v}_i(t)}{\sum_i m_i} \,.
\end{equation*}
\end{definition}

\begin{theorem}[Collective position and collective velocity]
Let $S$ be a free will system whose set of agents is $\{a_1, ..., a_n\}$, 
let $W$ be a discrete collective will for $S$,
and let $H$ be the discrete history of $S$ with respect to $W$.
Then
\begin{align*}
  \mathbf{x}_S(t + 1) = \mathbf{x}_S(t) + \mathbf{v}_S(t + 1) \,, &&\text{for all } t \in \mathbb{N} \,.
\end{align*}
\end{theorem}
\begin{proof}
By equation~\eqref{eq:position} we have
\begin{equation*}
  \mathbf{x}_i(t + 1) = \mathbf{x}_i(t) + \mathbf{v}_i(t + 1) \,,
\end{equation*}
from which it follows that
\begin{equation*}
  m_i \mathbf{x}_i(t + 1) = m_i \mathbf{x}_i(t) + m_i \mathbf{v}_i(t + 1) \,,
\end{equation*}
and therefore
\begin{equation*}
  \sum_i m_i \mathbf{x}_i(t + 1) = \sum_i m_i \mathbf{x}_i(t) + \sum_i m_i \mathbf{v}_i(t + 1) \,,
\end{equation*}
and 
\begin{equation*}
  \frac{\sum_i m_i \mathbf{x}_i(t + 1)}{\sum_i m_i} = \frac{\sum_i m_i \mathbf{x}_i(t)}{\sum_i m_i} + \frac{\sum_i m_i \mathbf{v}_i(t + 1)}{\sum_i m_i} \,,
\end{equation*}
which is equivalent to
\begin{equation*}
  \mathbf{x}_S(t + 1) = \mathbf{x}_S(t) + \mathbf{v}_S(t + 1) \,.
\end{equation*}
\end{proof}

\begin{theorem}[Conservation of collective velocity]
\label{th:conservation-velocity}
Let $S$ be a free will system whose set of agents is $\{a_1, ..., a_n\}$, 
let $W$ be a discrete collective will for $S$,
and let $H$ be the discrete history of $S$ with respect to $W$.
Then there exists a constant $\mathbf{v}_S \in \mathbb{R}^3$ such that
\begin{align*}
  \mathbf{v}_S = \mathbf{v}_S(t) \,, &&\text{for all } t \in \mathbb{N} \,.
\end{align*}
\end{theorem}
\begin{proof}
It is enough to notice that
\begin{equation*}
  \mathbf{v}_S(t) = \frac{\sum_i m_i \mathbf{v}_i(t)}{\sum_i m_i} 
  = \frac{\mathbf{p}_S(t)}{\sum_i m_i} 
  = \frac{\mathbf{p}_S}{\sum_i m_i} \,.
\end{equation*}
\end{proof}

\begin{definition}[Individual kinetic energy]
Let $S$ be a free will system whose set of agents is $\{a_1, ..., a_n\}$, 
let $W$ be a discrete collective will for $S$,
and let $H$ be the discrete history of $S$ with respect to $W$.
The \emph{individual kinetic energy} $K_i(t)$ of agent $a_i$ at instant of time $t$ is
\begin{equation*}
  K_i(t) = \frac{1}{2} m_i \lvert\lvert \mathbf{v}_i(t) \rvert\rvert^2 \,.
\end{equation*}
\end{definition}

\begin{definition}[Collective kinetic energy]
Let $S$ be a free will system whose set of agents is $\{a_1, ..., a_n\}$, 
let $W$ be a discrete collective will for $S$,
and let $H$ be the discrete history of $S$ with respect to $W$.
The \emph{collective kinetic energy} $K_S(t)$ of $S$ at instant of time $t$ is
\begin{equation*}
  K_S(t) = \frac{1}{2} \sum_i m_i \lvert\lvert \mathbf{v}_S(t) \rvert\rvert^2 \,.
\end{equation*}
\end{definition}

\begin{theorem}[Conservation of collective kinetic energy]
Let $S$ be a free will system whose set of agents is $\{a_1, ..., a_n\}$, 
let $W$ be a discrete collective will for $S$,
and let $H$ be the discrete history of $S$ with respect to $W$.
Then there exists a constant $K_S \in \mathbb{R}^+_0$ such that
\begin{align*}
  K_S = K_S(t) \,, &&\text{for all } t \in \mathbb{N} \,.
\end{align*}
\end{theorem}
\begin{proof}
The claim immediately follows from Theorem~\ref{th:conservation-velocity}.
\end{proof}

% ------------------------------------------
\section{Continuous time model of free will}
\label{se:continuous}

\begin{definition}[Continuous individual will]
\label{de:will-continuous}
Let $S$ be a free will system whose set of agents is $\{a_1, ..., a_n\}$.
A \emph{continuous individual will} of agent $a_i$ in $S$ is a set of functions
\begin{equation*}
  \left\{ \mathbf{F}_{ij} : \mathbb{R}^+_0 \to \mathbb{R}^3 \mid j \neq i \right\} \,,
\end{equation*}
such that, for each $j \neq i$ and for all $t \ge 0$, the integral
\begin{equation*}
  \int_0^t \mathbf{F}_{ij}(s) ds
\end{equation*}
exists and is finite.\footnote{Note that this integral is not a real number, but a vector in $\mathbb{R}^3$.}
\end{definition}

Intuitively, in the definition of continuous individual will, $\mathbf{F}_{ij}(t)$ is the force exterted by agent $a_i$ upon agent $a_j$ at
instant of time $t$, determined by the free will of agent $a_i$.

\begin{definition}[Continuous collective will]
Let $S$ be a free will system whose set of agents is $\{a_1, ..., a_n\}$.
A \emph{continuous collective will} for $S$ is the union $\bigcup_i w_i$, where $w_i$ is an individual will of agent $a_i$ in $S$, for all $i$.
\end{definition}

\begin{definition}[Continuous history]
Let $S$ be a free will system whose set of agents is $\{a_1, ..., a_n\}$, 
and let $W = \left\{ \mathbf{F}_{ij} : \mathbb{R}^+_0 \to \mathbb{R}^3 \mid i \neq j \right\}$ be a continuous collective will for $S$.
The \emph{continuous history} of $S$ with respect to $W$ is the following collection of functions:
\begin{itemize}
\item $\mathbf{F}_i : \mathbb{R}^+_0 \to \mathbb{R}^3$, for each $i$, where $\mathbf{F}_i(t)$ is the total force exterted on agent $a_i$ at instant of time $t$;

\item $\mathbf{a}_i : \mathbb{R}^+_0 \to \mathbb{R}^3$, for each $i$, where $\mathbf{a}_i(t)$ is the acceleration of agent $a_i$ at instant of time $t$;

\item $\mathbf{v}_i : \mathbb{R}^+_0 \to \mathbb{R}^3$, for each $i$, where $\mathbf{v}_i(t)$ is the velocity of agent $a_i$ at instant of time $t$;

\item $\mathbf{x}_i : \mathbb{R}^+_0 \to \mathbb{R}^3$, for each $i$, where $\mathbf{x}_i(t)$ is the position of agent $a_i$ at instant of time $t$.
\end{itemize}
For each $i$, the above functions are precisely defined as follows:
\begin{align} 
\label{eq:total-continuous}
  \mathbf{F}_i(t) = \sum_{j \neq i} \left( \mathbf{F}_{ji}(t) - \mathbf{F}_{ij}(t) \right) \,, &&\text{for all } t \ge 0 \,.
\end{align}

\begin{align}
\label{eq:second-continuous}
  \mathbf{a}_i(t) = \frac{\mathbf{F}_i(t)}{m_i} \,, &&\text{for all } t \ge 0 \,.
\end{align}

\begin{align}
\label{eq:velocity-continuous}
  \mathbf{v}_i(t) = \mathbf{v}_i^{(0)} + \int_0^t \mathbf{a}_i(s) ds \,, &&\text{for all } t \ge 0 \,.
\end{align}

\begin{align}
\label{eq:position-continuous}
  \mathbf{x}_i(t) = \mathbf{x}_i^{(0)} + \int_0^t \mathbf{v}_i(s) ds \,, &&\text{for all } t \ge 0 \,.
\end{align}
\end{definition}

Notice that equation~\eqref{eq:total-continuous} enforces agents to obey the Golden Rule.  
In fact, at each instant of time $t$, an agent $a_i$ that chooses to extert a force $\mathbf{F}_{ij}(t)$ upon agent $a_j$, must also 
extert the force $-\mathbf{F}_{ij}(t)$ upon itself.

% ------------------------------------------------------
\section{Continuous time properties}
\label{se:theorems-continuous}

\begin{theorem}[Balance of forces]
\label{th:zero-continuous}
Let $S$ be a free will system whose set of agents is $\{a_1, ..., a_n\}$, 
let $W$ be a continuous collective will for $S$,
and let $H$ be the continuous history of $S$ with respect to $W$.
Then, for each $i$, the following holds
\begin{align*}
  \sum_i \mathbf{F}_i(t) = 0 \,, &&\text{for all } t \ge 0 \,.
\end{align*}
\end{theorem}
\begin{proof}
The proof follows the same lines of the proof of Theorem~\ref{th:zero}.
\end{proof}

\begin{definition}[Individual momentum]
Let $S$ be a free will system whose set of agents is $\{a_1, ..., a_n\}$, 
let $W$ be a continuous collective will for $S$,
and let $H$ be the continuous history of $S$ with respect to $W$.
The \emph{individual momentum} $\mathbf{p}_i(t)$ of agent $a_i$ at instant of time $t$ is
\begin{equation*}
  \mathbf{p}_i(t) = m_i \mathbf{v}_i(t) \,.
\end{equation*}
\end{definition}

\begin{definition}[Collective momentum]
Let $S$ be a free will system whose set of agents is $\{a_1, ..., a_n\}$, 
let $W$ be a continuous collective will for $S$,
and let $H$ be the continuous history of $S$ with respect to $W$.
The \emph{collective momentum} $\mathbf{p}_S(t)$ of $S$ at instant of time $t$ is
\begin{equation*}
  \mathbf{p}_S(t) = \sum_i \mathbf{p}_i(t) \,.
\end{equation*}
\end{definition}

\begin{theorem}[Conservation of collective momentum]
Let $S$ be a free will system whose set of agents is $\{a_1, ..., a_n\}$, 
let $W$ be a continuous collective will for $S$,
and let $H$ be the continuous history of $S$ with respect to $W$.
Then there exists a constant $\mathbf{p}_S \in \mathbb{R}^3$ such that
\begin{align*}
  \mathbf{p}_S = \mathbf{p}_S(t) \,, &&\text{for all } t \ge 0 \,.
\end{align*}
\end{theorem}
\begin{proof}
It is enough to prove that $\mathbf{p}_S(t) = \mathbf{p}_S(0)$.  We have
\begin{align*}
  \mathbf{v}_i(t) &= \mathbf{v}_i^{(0)} + \int_0^t \mathbf{a}_i(s) ds \\
  &= \mathbf{v}_i(0) + \int_0^t \frac{\mathbf{F}_i(s)}{m_i} ds \,,
\end{align*}
which implies
\begin{equation*}
  m_i \mathbf{v}_i(t) = m_i \mathbf{v}_i(0) + \int_0^t \mathbf{F}_i(s) ds \,.
\end{equation*}
But then
\begin{equation*}
  \sum_i m_i \mathbf{v}_i(t) = \sum_i m_i \mathbf{v}_i(0) + \sum_i \int_0^t \mathbf{F}_i(s) ds \,.
\end{equation*}
Thus,
\begin{equation*}
  \mathbf{p}_S(t) = \mathbf{p}_S(0) + \int_0^t \sum_i \mathbf{F}_i(s) ds \,,
\end{equation*}
and the claim follows by Theorem~\ref{th:zero-continuous}.
\end{proof}

\begin{definition}[Collective position]
Let $S$ be a free will system whose set of agents is $\{a_1, ..., a_n\}$, 
let $W$ be a continuous collective will for $S$,
and let $H$ be the continuous history of $S$ with respect to $W$.
The \emph{collective position} $\mathbf{x}_S(t)$ of $S$ at instant of time $t$ is
\begin{equation*}
  \mathbf{x}_S(t) = \frac{\sum_i m_i \mathbf{x}_i(t)}{\sum_i m_i} \,.
\end{equation*}
\end{definition}

\begin{definition}[Collective velocity]
Let $S$ be a free will system whose set of agents is $\{a_1, ..., a_n\}$, 
let $W$ be a continuous collective will for $S$,
and let $H$ be the continuous history of $S$ with respect to $W$.
The \emph{collective velocity} $\mathbf{v}_S(t)$ of $S$ at instant of time $t$ is
\begin{equation*}
  \mathbf{v}_S(t) = \frac{\sum_i m_i \mathbf{v}_i(t)}{\sum_i m_i} \,.
\end{equation*}
\end{definition}

\begin{theorem}[Collective position and collective velocity]
Let $S$ be a free will system whose set of agents is $\{a_1, ..., a_n\}$, 
let $W$ be a continuous collective will for $S$,
and let $H$ be the continuous history of $S$ with respect to $W$.
Then
\begin{align*}
  \mathbf{x}_S(t) = \mathbf{x}_S(0) + \int_0^t \mathbf{v}_S(s) ds \,, &&\text{for all } t \ge 0 \,.
\end{align*}
\end{theorem}
\begin{proof}
We have
\begin{equation*}
  \mathbf{x}_i(t) = \mathbf{x}_i^{(0)} + \int_0^t \mathbf{v}_i(s) ds \,.
\end{equation*}
from which it follows that
\begin{equation*}
  m_i \mathbf{x}_i(t) = m_i \mathbf{x}_i(0) + \int_0^t m_i \mathbf{v}_i(s) ds \,,
\end{equation*}
and therefore
\begin{equation*}
  \sum_i m_i \mathbf{x}_i(t) = \sum_i m_i \mathbf{x}_i(0) + \int_0^t \sum_i m_i \mathbf{v}_i(s) ds \,,
\end{equation*}
and 
\begin{equation*}
  \frac{\sum_i m_i \mathbf{x}_i(t)}{\sum_i m_i} = \frac{\sum_i m_i \mathbf{x}_i(0)}{\sum_i m_i} + \int_0^t \frac{\sum_i m_i \mathbf{v}_i(s)}{\sum_i m_i} ds \,,
\end{equation*}
which is equivalent to
\begin{equation*}
  \mathbf{x}_S(t) = \mathbf{x}_S(0) + \int_0^t \mathbf{v}_S(s) ds \,.
\end{equation*}
\end{proof}

\begin{theorem}[Conservation of collective velocity]
\label{th:conservation-velocity-continuous}
Let $S$ be a free will system whose set of agents is $\{a_1, ..., a_n\}$, 
let $W$ be a continuous collective will for $S$,
and let $H$ be the continuous history of $S$ with respect to $W$.
Then there exists a constant $\mathbf{v}_S \in \mathbb{R}^3$ such that
\begin{align*}
  \mathbf{v}_S = \mathbf{v}_S(t) \,, &&\text{for all } t \ge 0 \,.
\end{align*}
\end{theorem}
\begin{proof}
It is enough to notice that
\begin{equation*}
  \mathbf{v}_S(t) = \frac{\sum_i m_i \mathbf{v}_i(t)}{\sum_i m_i} 
  = \frac{\mathbf{p}_S(t)}{\sum_i m_i} 
  = \frac{\mathbf{p}_S}{\sum_i m_i} \,.
\end{equation*}
\end{proof}

\begin{definition}[Individual kinetic energy]
Let $S$ be a free will system whose set of agents is $\{a_1, ..., a_n\}$, 
let $W$ be a continuous collective will for $S$,
and let $H$ be the continuous history of $S$ with respect to $W$.
The \emph{individual kinetic energy} $K_i(t)$ of agent $a_i$ at instant of time $t$ is
\begin{equation*}
  K_i(t) = \frac{1}{2} m_i \lvert\lvert \mathbf{v}_i(t) \rvert\rvert^2 \,.
\end{equation*}
\end{definition}

\begin{definition}[Collective kinetic energy]
Let $S$ be a free will system whose set of agents is $\{a_1, ..., a_n\}$, 
let $W$ be a continuous collective will for $S$,
and let $H$ be the continuous history of $S$ with respect to $W$.
The \emph{collective kinetic energy} $K_S(t)$ of $S$ at instant of time $t$ is
\begin{equation*}
  K_S(t) = \frac{1}{2} \sum_i m_i \lvert\lvert \mathbf{v}_S(t) \rvert\rvert^2 \,.
\end{equation*}
\end{definition}

\begin{theorem}[Conservation of collective kinetic energy]
Let $S$ be a free will system whose set of agents is $\{a_1, ..., a_n\}$, 
let $W$ be a continuous collective will for $S$,
and let $H$ be the continuous history of $S$ with respect to $W$.
Then there exists a constant $K_S \in \mathbb{R}^+_0$ such that
\begin{align*}
  K_S = K_S(t) \,, &&\text{for all } t \ge 0 \,.
\end{align*}
\end{theorem}
\begin{proof}
The claim immediately follows from Theorem~\ref{th:conservation-velocity-continuous}.
\end{proof}

% -------------------------------------
\section{Comparison with Newton's laws}
\label{se:newton}

\subsection{Newton's first law}
Newtont's first law of motion states that a body continues in its state of rest, or in uniform motion in a straight line, unless acted upon by a force.
Mathematically, this law can be expressed as
\begin{equation*}
  \mathbf{F} = 0 \iff \mathbf{a} = 0 \,,
\end{equation*}
where $\mathbf{F}$ is the net force acting upon the body, and $\textbf{a}$ is the acceleration of the body.
In our model of free will, where bodies are replaced by agents, Newton's first law holds clearly as a consequence of equation~\eqref{eq:second}
in the discrete case, and of equation~\eqref{eq:second-continuous} in the continuous case.

\subsection{Newton's second law}
Newton's second law of motion states that a body acted upon by a force moves with acceleration directly proportional to the intensity of the force, and inversely
proportional to its mass, with the acceleration being in the same direction of the force.
Mathematically, this law can be expressed as
\begin{equation*}
  \mathbf{F} = m \mathbf{a} \,,
\end{equation*}
where $\mathbf{F}$ is the net force acting upon the body, $m$ is the mass of the body, and $\textbf{a}$ the acceleration of the body.
In our model of free will, Newton's second law holds clearly as a consequence of equation~\eqref{eq:second}
in the discrete case, and of equation~\eqref{eq:second-continuous} in the continuous case.

\subsection{Newton's third law}
Newton's third law of motion states that if two bodies exert forces on each other, these forces are equal in intensity and opposite in direction.
Mathematically, this law can be expressed as
\begin{equation}
\label{eq:third}
  \mathbf{F}_A = - \mathbf{F}_B \,,
\end{equation}
where $\mathbf{F}_A$ is the force exterted by body $A$ on body $B$, and $\mathbf{F}_B$ is the force exterted by body $B$ on body $A$.
In our model of free will, Newtons's third law still holds, but provided it is given a conceptually different interpretation.

In Definition~\ref{de:will}, we denoted with $\mathbf{F}_{ij}$ the force exterted by agent $a_i$ upon agent $a_j$, determined by the free will of agent $a_i$.
Conversely, we denoted with $\mathbf{F}_{ji}$ the force exterted by agent $a_j$ upon agent $a_i$, determined by the free will of agent $a_j$.
Now, the equation $\mathbf{F}_{ij} = -\mathbf{F}_{ji}$ does not need to hold.  Does this mean that in our model, Newton's third law is false?
Not necessarily.  Let us see why.

By equation~\eqref{eq:total} in the discrete case or by equation~\eqref{eq:total-continuous} in the continuous case, 
when agent $a_i$ has the power to extert a force $\mathbf{F}_{ij}$ upon agent $a_j$, it must pay
the price of having to extert upon itself the force $-\mathbf{F}_{ij}$.
Now, consider a finite free will system $S = \{a_1, a_2\}$ with only two agents.  Then, by equation~\eqref{eq:total} or equation~\eqref{eq:total-continuous}, we have
\begin{equation*}
  \mathbf{F}_1 = \mathbf{F}_{21} - \mathbf{F}_{12} \,,
\end{equation*}
and
\begin{equation*}
  \mathbf{F}_2 = \mathbf{F}_{12} - \mathbf{F}_{21} \,,
\end{equation*}
from which it follows
\begin{equation}
\label{eq:us}
  \mathbf{F}_1 = - \mathbf{F}_2 \,.
\end{equation}
Equation~\eqref{eq:us} is the analogous of equation~\eqref{eq:third}.  However, note that $\mathbf{F}_1$ is not the force extered by agent $a_2$ on agent $a_1$,
but is instead the net force exterted on agent $a_1$ in a finite free will system comprising only two agents.  Similarly, $\mathbf{F}_2$ is not the force
exterted by agent $a_1$ on agent $a_2$, but is the net force exterted on agent $a_2$ in a finite free will system comprising only two agents.  
It follows that Newton's third law becomes true in our model of free will when the number of agents in a finite free will system is equal to $2$, provided the
meaning of $\mathbf{F}_1$ and $\mathbf{F}_2$ is interpreted adequately.

Let us discuss some practical cases of our different interpretation of Newton's third law.  When a person walks, the common interpration
is that the feet push against the floor, and the floor pushes against the feet.  In our interpretation, instead, the feet push against the floor,
and the opposite force of equal intensity that pushes from the floor against the feet is generated by the feet themselves!  Yes, there are two opposite
forces of equal intensity, but they are both generated by the feet, and not by the floor.

Now, consider the case of two fists colliding each other.  In this case, there are two agents.  Let $a_1$ be the first fist, and let $a_2$ be the second fist.
In this case there are four forces involved:
\begin{itemize}
\item the force $\mathbf{F}_{12}$ that the first fist exterts on the second, decided by free will;

\item the opposite force $-\mathbf{F}_{12}$ that the first fist exterts on itself, and that appears to come from the second fist;

\item the force $\mathbf{F}_{21}$ that the second fist exterts on the first, decided by free will;

\item the opposite force -$\mathbf{F}_{21}$ that the second fist exterts on itself, and that appears to come from the first fist.
\end{itemize} 
Now, it does not need to be the case that $\mathbf{F}_{12} = -\mathbf{F}_{21}$, since any fist may hit with more force than the other.
However, it must be the case that $\mathbf{F}_1 = -\mathbf{F}_2$, where $\mathbf{F}_1$ is the net force exterted on the first fist, and
$\mathbf{F}_2$ is the net force exterted on the second fist.

% -----------------------------
\section{When choices are made}
\label{se:timing}

A subtle point of our model of free will is when agents construct their individual wills.  
From Definition~\ref{de:will} in the discrete case or Definition~\ref{de:will-continuous} in the continuous case, it may seem that 
agents choose their individual wills before time begins, but if this were the case, then the histories of free will systems would
be completely deterministic.
Instead, agents choose their indivisual wills gradually, as time goes by.  

Precisely, in both the discrete and continuous cases, the values $\mathbf{F}_{ij}(t)$ are chosen by the agents exactly at instant of time $t$,
after observing the positions $\mathbf{x}_i(t)$ and velocities $\mathbf{v}_i(t)$ at instant of time $t$.
We also remark that, in the discrete case the positions $\mathbf{x}_i(t)$ and velocities $\mathbf{v}_i(t)$ at instant of time $t$
are completely defined by the history of the system up to instant of time $t - 1$.  Instead, in the continous case the 
positions $\mathbf{x}_i(t)$ and velocities $\mathbf{v}_i(t)$ at instant of time $t$ are completely defined by the history
of the system in the half open interval $[0, t)$.  We finally notice that every discrete or continuous history up to
instant of tiem $t$ is not influenced by the individual wills chosen by the agents at those instants of time $t' > t$.
In other words, the future does not influence the past, nor does influence the present moment.

% -------------------------
\section{Bounded free will}
\label{se:bounded}

In our model of free will, agents must observe the Golden Rule, as enforced by equation~\eqref{eq:total} in the discrete case and equation~\eqref{eq:total-continuous} in the continuous case.
But apart from this obligation, agents are quite powerful.  They have in fact the power to extert their forces on each other, no matter what is the distance between them.  
In other words, agents' interactions are non-local.
Furthermore, our Definition~\ref{de:will} of discrete individual will is quite liberal, since the functions $\mathbf{F}_{ij} : \mathbb{N} \to \mathbb{R}^3$ can be just about anything.
Similarly, our Definition~\ref{de:will-continuous} of continuous individual free will is also quite liberal, since apart from the existence of an integral, it virtually
does not pose any restriction on the functions $\mathbf{F}_{ij} : \mathbb{R}^+_0 \to \mathbb{R}^3$.  One may indeed wonder whether all this power is excessive.

To reduce the power of agents, some restrictions on free will may be imposed.  The following definition aims to reduce the non-locality of interactions between agents, as well
as the intensity of the interactions between agents.

\begin{definition}[Bounded free will]
Let $S$ be a free will system whose set of agents is $\{a_1, ..., a_n\}$, 
let $W$ be a discrete (respectively, continuous) collective will for $S$,
and let $H$ be the discrete (respectively, continuous) history of $S$ with respect to $W$.
We say that $W$ is \emph{bounded} for $S$ if there exists a constant $\lambda > 0$ and a constant $\epsilon > 0$ such that, for all $i, j$, we have
\begin{align*}
  \lvert\lvert \mathbf{F}_{ij}(t) \rvert\rvert \le \frac{\lambda}{\lvert\lvert \mathbf{x}_i(t) - \mathbf{x}_j(t) \rvert\rvert + \epsilon} \,, &&\text{for all } t \ge 0
\end{align*} 
\end{definition}

According to the definition of bounded free will, the intensities of the forces $\mathbf{F}_{ij}$ that agents extert on each other are limited by $\frac{\lambda}{\epsilon}$.
Notice also that, when agents are far apart from each other, their distance further limits the intensity of their interactions, thus reducing non-locality.

\begin{example}
\label{ex:bounded}

Let $S = \{a_1, a_2\}$ be the free will system specified in Example~\ref{ex:system}.  
Let us specify the following bounded collective will $W$
\begin{align*}
  \mathbf{F}_{12}(t) &= \left( \frac{2}{\lvert\lvert \mathbf{x}_i(t) - \mathbf{x}_j(t) \rvert\rvert + 1}, 0, 0 \right) \,, &\text{for all } t \ge 0 \,, \\
  \mathbf{F}_{21}(t) &= \left( \frac{-1}{\lvert\lvert \mathbf{x}_i(t) - \mathbf{x}_j(t) \rvert\rvert + 1}, 0, 0 \right) \,, &\text{for all } t \ge 0 \,.
\end{align*}
Using a program written in Java, we run a simulation of the discrete history of $S$ with respect to $W$.
If the simulation is run for $10^6$ iterations, the final positions and velocities computed are the following
\begin{verbatim}
  t = 1000000
    x[1] = (-6.54865e+06, 0.0, 0.0)
    v[1] = (-6.78726, 0.0, 0.0)
    x[2] = (6.54865e+06, 0.0, 0.0)
    v[2] = (6.78726, 0.0, 0.0)
\end{verbatim}
A closer look at the simulation shows that, while the positions $\mathbf{x}_i(t)$ diverge, the velocities $\mathbf{v}_i(t)$ converge.
In particular, the values $\lvert \lvert \mathbf{v}_i(t) \rvert\rvert$ are limited and converge to a maximum value.
\end{example}

Example~\ref{ex:bounded} suggests that free will systems subject to bounded collective wills have histories in which there
is a maximum velocity allowed.  This observation leads us to the following conjecture.

\begin{conjecture}
\label{co:light}

Let $S$ be a free will system whose set of agents is $\{a_1, ..., a_n\}$, 
let $W$ be a discrete (respectively, continuous) bounded collective will for $S$,
and let $H$ be the discrete (respectively, continuous) history of $S$ with respect to $W$.
Then there exists a constant $c > 0$ such that, for each i
\begin{align*}
  \lvert \lvert \mathbf{v}_i(t) \rvert\rvert \le c \,, && \text{for all } t \ge 0 \,.
\end{align*}
\end{conjecture}

We do not have yet a formal proof of Conjecture~\ref{co:light}, but only simulations that hint at it.
If the conjecture were proved to be true, it would mean that bounded free will limits not only the intensities
of the forces between the agents, but also the maximum velocity of agents.

% ------------------
\section{Conclusion}
\label{se:conclusion}
In order to study the concept of free will, we have extended Newton's laws of motion with a mathematical
model in which bodies are replaced with agents endowed with free will.  
In our model, the free will of agents is bound by the
Golden Rule, an ethic found in most cultures which states that one should wish upon others as one wishes
upon himself.  Formally, our model contains the following rule: whenever an agent exterts a force on another
agent, it must also extert upon itself a force of equal intensity and opposite direction.

Our model of free will requires that the number of involved agents is finite.  It would be nice to extend it
to an infinite number of agents, although it is not clear how such an extension would work, since it would
need to handle the presence of infinite sums.

Our model of free will specifies the limits and rules that the choices of agents must abide, but it does
not explain how these choices are made.  Indeed, the existence of free will seems to imply that the choices 
of agents cannot be random, and must instead rely on some form of intelligence which is hidden and immaterial.

Is our model of free will a model of reality?  That is hard to say, because of inherent difficulties in 
verifying it empirically.  Indeed, free will seems to be a property of complex beings, whereas in
our model agents are not complex beings but points.  Do points have free will in reality?  We do not know.
On the other hand, the closer approximation in reality to our concept of point is the concept of elementary particle, and elementary
particles may have free will, as stated by the free will theorem of Conway and Kochen.

\bibliography{will}

\end{document}